%% file: main-arxiv.tex
\documentclass[a4paper]{article}

\usepackage{etex} 
\reserveinserts{28}

\usepackage{xcolor}
\include{header}

\usepackage{latexsym,amsthm,amsmath,amssymb,url}

\usepackage{tikz,authblk}
\usepackage{skak}

\graphicspath{{./pdfs/}}

\title{Lower and Upper Bound for Computing the Size of All Second Neighbourhoods}

\author[1]{Gregory Gutin\thanks{Partially supported by the Royal Society Wolfson Research Merit Award.}}
\author[2]{George B. Mertzios\thanks{Partially supported by the EPSRC Grant~EP/P020372/1.}}
\author[1]{Felix Reidl}

\affil[1]{Department of Computer Science, Royal Holloway, University of London, UK. \newline
Email: \texttt{g.gutin@rhul.ac.uk}, \texttt{felix.reidl@rhul.ac.uk}}
\affil[2]{Department of Computer Science, Durham University, UK. \newline
Email: \texttt{george.mertzios@durham.ac.uk}}

\begin{document}

\def\eps{\epsilon}
\def\SAT{\Problem{Satisfiability}}

\maketitle
\begin{abstract}
  We consider the problem of computing the size of each $r$-neighbourhood 
  for every vertex of a graph. Specifically, we ask whether the size
  of the closed second neighbourhood can be computed in subquadratic time.
  
  Adapting the SETH reductions by Abboud et al. (2016) that exclude
  subquadratic algorithms to compute the radius of a graph, we find that a
  subquadratic algorithm would violate the SETH. On the other hand, a linear
  fpt-time algorithm by Demaine et al. (2014)
  parameterized by a certain `sparseness parameter' of the graph is known,
  where the dependence on the parameter is exponential. We show here that a
  better dependence is unlikely: for any~$\delta < 2$, no algorithm running
  in time $O(2^{o(\vc(G))} \, n^{\delta})$,
   where~$\vc(G)$ is the vertex cover
  number, is possible unless the SETH fails.

  We supplement these lower bounds with algorithms that solve the problem in
  time~$O(2^{\vc(G)/2} \vc(G)^2 \cdot n)$ and $O(2^w w \cdot n)$.
\end{abstract}

%
%
\section{Introduction}

For a vertex $v$ of a graph $G$ and an integer $r\ge 1$,
$N^r(v)$ ($N^r[v]$, respectively) denotes the set of vertices of $G$ of
distance exactly (at most, respectively) $r$ from $v$. For a graph $G$, $|G|$
will denote the number of vertices in $G$. As usual in graph algorithms
literature, unless defined differently, $n$ and $m$ will denote the number of
vertices and edges in the input graph.
In this paper, we consider the following two basic problems on graphs.

\begin{problem}{$r$-Neighbourhood Sizes}
	\Input & A graph~$G$ and an integer~$r$. \\
	\Prob & Compute for every vertex~$v \in G$ the size of~$N^r(v)$.
\end{problem}

\begin{problem}{Closed $r$-Neighbourhood Sizes}
	\Input & A graph~$G$ and an integer~$r$. \\
	\Prob & Compute for every vertex~$v \in G$ the size of~$N^r[v]$.
\end{problem}

\noindent
Since both problems are easily Turing-reducible to each other, we 
focus on the closed neighbourhood variant in the following with the understanding that
all results transfer to the open neighbourhood variant.
Without loss of generality, we will assume in the remainder of the paper that
the input graph is connected. Clearly we can solve the above problems in
time~$O(n(m+n))$ by conducting a (truncated) breadth-first search from every
vertex.  This means $\Omega(n^2)$ time even for sparse connected graphs. The
following question is natural: can we solve {\sc Closed $r$-Neighbourhood
Sizes} in a subquadratic (in $n$) time even for $r=2$? We will show in Theorem
\ref{lemma:no-subquadratic} that this is not possible provided the Strong
Exponential Time Hypothesis (SETH) holds. SETH has been put forward by
Impagliazzo and Paturi~\cite{SETH}, stating that, for every positive~$\eps
<1$, there exists an integer~$r$ such that \Problem{$r$-CNF SAT} cannot be
solved in time~$O(2^{\eps n})$, where $n$ is the number of variables in the
input $r$-CNF formula. More precisely, define~$s_r$ to be the infimum over all
numbers~$\delta$ for which there exists an algorithm that
solves~\Problem{$r$-CNF SAT} in time~$2^{\delta n} (n+m)^{O(1)}$. The
exponential time hypothesis~(ETH) states that~$s_3 > 0$, that is, there is no
subexponential algorithm solving~\Problem{3SAT}. SETH asserts that the limit
of the sequence~$(s_r)_{r \in \mathbb N}$ is 1.

Since subquadratic algorithms seem to be out of reach for
\Problem{Closed $2$-Neighbourhood Sizes}, we ask whether we can trade-off some
of the polynomial complexity in the input size for an exponential dependence
on some structural parameter of the input graph. Demaine~\etal showed that a
running time of~$O(2^{\dir \Delta_r(G)} n)$ is indeed possible (for the general
\Problem{Closed $r$-Neighbourhood Sizes} problem) where~$\dir \Delta_r$ is a certain
measure of the sparsity of~$G$~\cite{SparseComplexNetworks} which we describe briefly
below. Without going
into further detail here, we note that~$\dir \Delta_r$ satisfies~$\dir \Delta_2(G)
\leq \vc(G)$, where $\vc(G)$ is the minimum size of a vertex cover of $G$,
i.e.~a set which contains at least one vertex of every edge of $G$. Can we use
the following trade-off in the running time of~\cite{SparseComplexNetworks}:
replace $n$ by a subquadratic function in $n$ and replace $\dir \Delta_2(G)$
by $o(\vc(G))$~? We prove in Theorem  \ref{lemma:vc-lower-bound} that the
answer to this question is negative, assuming SETH. Therefore, since the
parameters treewidth and tree-depth\footnote{We define these two parameters
in the next section.} are smaller than the vertex cover number, the same
impossibility result follows also if we replace $vc(G)$ by any of these
parameters, see Corollary \ref{cor:variouspar}.

In contrast, we show in Theorem \ref{theorem:vc-upper-bound} that {\sc Closed
$2$-Neighbourhood Sizes} can even be solved in linear time in $n$ if the
factor is exponential in $\vc(G)/2$ (where the base of the exponent is 2). In
Theorem \ref{theorem:tw-upper-bound}, we prove that the same result is true if
we replace $\vc(G)/2$ by treewidth.

%
%
\section{Preliminaries}

In the following we will 
make explicit use of the
\emph{sparsification lemma} by Calabro, Impagliazzo, and
Paturi~\cite{Sparsification}:

\begin{lemma}[Sparsification Lemma~\cite{Sparsification}]
  For every~$r \in \N$ and~$\eps > 0$ there exists an algorithm which, given
  an~$r$-CNF formula~$\phi$ over~$n$ variables, outputs in time~$2^{\eps n} n^{O(1)}$
  a list of $r$-CNF formulas~$(\psi_i)_{i \leq t}$, where $t \leq 2^{\eps n}$, such that
  \begin{itemize}
    \item $\phi$ is satisfiable if and only if at least on~$\psi_i$ is satisfiable and
    \item each formula~$\psi_i$ has at most~$n$ variables, each of which occurring at most~$O((\frac{r}{\eps})^{3r})$ times.
  \end{itemize}
\end{lemma}

\noindent
We now formally define the notions of a \emph{tree decomposition} and of a \emph{nice tree decomposition}, which are key to our analysis below.

\begin{definition}  
	Given a graph $G = (V,E)$, a \emph{tree decomposition} of
	$G$ is a pair $({\cal T}, \beta)$, where ${\cal T}$ is a tree and
	$\beta:V({\cal T}) \rightarrow 2^V$ such that  $\bigcup_{x \in V({\cal
	T})}\beta(x) = V$, for each edge $uv \in E$, there exists a node $x \in
	V({\cal T})$ such that $u,v \in \beta(x)$, and for each $v \in V$, the set
	$\beta^{-1}(v)$ of nodes  form a connected subgraph (i.e.~a subtree) in ${\cal T}$.

	The \emph{width} of $({\cal T}, \beta)$ is $\max_{x \in V({\cal T})}(|\beta(x)|-1)$.
	The \emph{treewidth} of $G$ (denoted by $\tw(G)$) is the minimum width of all tree decompositions of $G$.
\end{definition}

\noindent
A \emph{path decomposition} of $G$ is defined similar to a \emph{tree
decomposition} of $G$, but the only trees $\cal T$ allowed are paths. This
leads to the \emph{pathwidth} of $G$ denoted $\pw(G)$.

 

\begin{definition}
\label{nice-tree-decom-def}
Given an undirected graph $G = (V,E)$, a \emph{nice tree decomposition} $({\cal T}, \beta)$ is a tree decomposition such that ${\cal T}$ is a rooted tree, and each of the nodes $x \in V({\cal T})$ falls under one of the following classes:
 \begin{itemize}
  \item {\bf $x$ is a Leaf node:} then $x$ has no children in ${\cal T}$;
  \item {\bf $x$ is an Introduce node:} then $x$ has a single child $y$ in ${\cal T}$, and there exists a vertex $v \notin \beta(y)$ such that $\beta(x) = \beta(y) \cup \{v\}$;
  \item {\bf $x$ is a Forget node:} then $x$ has a single child $y$ in ${\cal T}$, and there exists a vertex $v \in \beta(y)$ such that $\beta(x) = \beta(y) \setminus \{v\}$;
  \item {\bf $x$ is a Join node:} then $x$ has exactly two children $y$ and $z$, and $\beta(x) = \beta(y) = \beta(z)$.
 \end{itemize}
\end{definition}
 
\noindent
It is well-known \cite{kloks1994} that any given tree decomposition of a graph
can be transformed into a nice tree decomposition of the same width in
polynomial time.

For a rooted tree~${\cal T}$ and a node~$i \in {\cal T}$ we will write~${\cal
T}_i$ to denote the subtree of~$\cal T$ which includes~$i$ and all its
descendants. We consider~${\cal T}_i$ to be rooted in~$i$.

Besides width-measures like treedepth, pathwith, and treewidth we will further
consider the \emph{sparseness} parameters $\grad_1$ and~$\topgrad_1$ (see
Definition~\ref{grad-topgrad-def}). Recall that a graph~$H$ is a \emph{minor}
of a graph~$G$ if $H$ can be obtained from~$G$ by contracting a collection of
disjoint connected subgraphs and then taking a (not necessarily induced)
subgraph. If we impose the restriction that each contracted subgraph further
has radius at most~$r$ (that is, there exists a vertex in it from which every
other vertex has distance at most~$r$ within the subgraph), then we say
that~$H$ is an \emph{$r$-shallow minor} of~$G$ and we write~$H \sminor^r G$.

Recall that~$H$ is a \emph{topological minor} of a graph~$G$ if we can select
$|V(H)|$ vertices in~$G$ (the \emph{nails}) and connect them by~$|E(H)|$
internally vertex-disjoint paths~$\mathcal P$ such that~$uv \in H$ if and only
if the corresponding nails~$u', v'$ in~$G$ are connected by a path
in~$\mathcal P$. If we further impose the restriction that all paths
in~$\mathcal P$ have length at most~$2r+1$, then $H$ we say that~$H$ is an
\emph{$r$-shallow topological minor} of~$G$ and we write~$H \stminor^r G$.
Note that every $r$-shallow topological minor is in particular an $r$-shallow
minor.

With these two containment notions, we can now define the sparseness
parameters~$\grad_r$ and~$\topgrad_r$. For a more in-depth introduction
to the topic of shallow minors we refer to the book by \Nesetril and
Ossona de Mendez~\cite{Sparsity}.

\begin{definition}[Grad and top-grad]
\label{grad-topgrad-def}
  For a graph~$G$ and an integer~$r \geq 0$, we define the \emph{greatest
    reduced average density (grad) at depth~$r$} as
  \[
    \grad_r(G) = \max_{H \sminor^r G}  \frac{|E(H)|}{|V(H)|}
  \]
  and the the \emph{topologically greatest
    reduced average density (top-grad) at depth~$r$} as
  \[
    \topgrad_r(G) = \max_{H \stminor^r G}  \frac{|E(H)|}{|V(H)|}.
  \]
\end{definition}

\noindent
The following is a simple observation relevant to our results below:

\begin{observation}\label{obs:nabla-vc}
	$\topgrad_1(G) \leq \grad_1(G) \leq \vc(G)$.
\end{observation}
\begin{proof}
	The first inequality follows immediately since every $1$-shallow
	topological minor is also a $1$-shallow minor. To prove the second
	inequality, let~$X \subseteq V(G)$ be a minimal vertex cover of~$G$ and
	let~$H$ be a $1$-shallow minor of~$G$ (i.e.~$H \sminor^1 G$) with~$\grad_0(H) =
	|E(H)|/|V(H)| = \grad_1(G)$. Let us choose~$H$ among all minors that
	satisfy this relation such that~$|V(G)|$ is minimal.

	Contracting an $1$-shallow minor is equivalent to contracting a star 
	forest. Let~$\{S_x\}_{x \in H}$ be the stars contracted to obtain~$H$,
	identified by the resulting vertex in~$H$. Note that every star~$S_x$ with
	more than one vertex necessarily intersects with~$X$, thus the number
*	of such stars is~$|\{S_x \mid S_x \cap X \neq \emptyset \}_{x \in H}| \leq |X|$.
	Let us call these stars \emph{big} and all other stars \emph{small}.
	Note that a small star is simply a single vertex in~$V(G)\setminus X$. It follows
	that for two small stars~$S_x, S_y$ we have that~$xy \not \in H$.

	If~$\grad_0(H) \leq |X|$ we are done, thus assume that~$\grad_0(H) > |X|$. 
	By minimality, it follows that the minimum degree $\delta(H)$ of $H$ satisfies~$\delta(H) \geq \grad_0(H) > |X|$, otherwise
	we could remove a vertex of minimal degree without decreasing~$\grad_0(H)$,
	contradicting our choice of~$H$. But then there cannot be any small stars
	since their corresponding vertices have degree at most~$|X|$. We arrive
	at a contradiction since then only~$|X|$ vertices remain in~$H$, making
	a density of~$|X|$ impossible.
\end{proof}

\noindent
Note that the bound of Observation~\ref{obs:nabla-vc} is asymptotically tight.
Indeed, consider the graph~$K_{s,t}$; clearly, if we keep~$t$ fixed and
let~$s$ grow, the density approaches~$t = \vc(K_{s,t})$ from below.
Furthermore note that taking a $1$-shallow minor does not affect this argument
(for example, adding all edges to the side of size~$t$ does not improve
asymptotic bound).

Let us now discuss the parameter~$\dir \Delta_r(G)$. It is defined as the
maximum in-degree of so-called \emph{transitive fraternal augmentations}
of~$G$. The first augmentation~$\dir G_1$ is simply an acyclic augmentation
that minimizes the maximum in-degree; $\dir G_i$ is then computed from
$\dir G_{i-1}$ by the following two rules:
\begin{enumerate}
  \item If~$uv, vw \in \dir G_{i-1}$ then~$uw \in \dir G_i$; and
  \item if~$uv, wv \in \dir G_{i-1}$ then either~$uw \in \dir G_i$
        \emph{or} $wu \in \dir G_i$.
\end{enumerate}
The orientation in the second case is chosen such that~$\dir G_i$
has the smallest possible maximum in-degree. The following lemma illucidats the
realtionship between dtf-augmentations and vertex covers. We have to phrase it slighlty
weaker than the bounds on~$\grad_1$ and~$\topgrad_1$ since the value of~$\dir \Delta_r$
depends on how the augmentation was computed.

\begin{observation}\label{obs:dtf-vc}
  There exists a dtf-augmentation of~$G$ with
  $\dir \Delta_r(G) \leq \vc(G)$ for all~$r \geq 1$.
\end{observation}
\begin{proof}
  Let~$X$ be a minimal vertex cover of~$G$ and let~$Y := V(G)\setminus X$. We
  construct~$\dir G_1$ by orienting all edges incident to~$Y$ towards~$Y$ and
  choose an arbitrary acyclic orientation for all other edges.

  Any augmentation built from~$\dir G_1$ will not add any out-arcs to~$Y$;
  thus the maximum in-degree of vertices in~$Y$ is~$|X|$. The maximum in-degree
  of a vertex in~$X$ is~$|X|-1$ since no arc will point from~$Y$ to~$X$. This proves
  the claim.
\end{proof}

%
%
\section{Lower Bounds}

We adapt the construction of Abboud, Williams, and Wang \cite{Abboud2016} to
verify our intuition that computing
neighbourhood sizes is probably not possible in subquadratic time. 
\begin{theorem}\label{lemma:no-subquadratic}
	For any $\eps >0$, \Problem{Closed $2$-Neighbourhood Sizes} on a graph $G$ cannot be solved in time~$O(|G|^{2-\eps})$, unless SETH fails.
\end{theorem}
\begin{proof}
	Consider a \SAT instance~$\phi$ with variables~$x_1,\ldots,x_n$ and a set $C$ of $m$ clauses. For
	simplicity, assume that~$n$ is even and define~$N := 2^{n/2}$. We partition the
	variables into sets~$X_l := \{x_1,\ldots,x_{n/2}\}$ and~$X_h :=
	\{x_{n/2+1},\ldots,x_n\}$.
	
	Now construct a graph~$G$ as follows: create one vertex for each
	of the~$N=2^{n/2}$ possible truth assignments~$\alpha_i\colon X_l \to
	\{0,1\}^{n/2}$ of~$X_l$; call the set of these vertices~$A :=
	\{\alpha_i\}_{i \in [N]}$. Proceed similarly for~$X_h$ and create a set~$B$ of~$N$ vertices, corresponding to all truth assignments~$\beta_i\colon X_h \to\{0,1\}^{n/2}$.
	Furthermore create one vertex~$c_i$ for every clause in~$\phi$ and call by~$C$ the resulting set of $m$ vertices. Finally, create two additional vertices~$v_a,v_b$. 
	
  Now for every partial assignment~$\gamma \in A \cup B$, connect~$\gamma$ to 
  each clause~$c \in C$ which is \emph{not} satisfied by~$\gamma$ (we consider
  a clause to be satisfied under a partial assignment if at least one positive
  variable of the clause is set to true or at least one negative variable is
  set to false). Finally, connect vertex~$v_a$ to vertex~$v_b$ and to all vertices in~$A \cup C$, and connect vertex~$v_b$ to 
  all vertices in~$B \cup C$. This concludes the construction of~$G$, which can be executed in~$O(N m)$ time. For an illustration of this construction see Figure~\ref{figure-quadratic-lower-bound}.
	
	Note that, if there exist partial truth assignments~$\alpha \in A$ and~$\beta \in B$ such
	that~${N[\alpha] \cap N[\beta] = \emptyset}$, then the
	truth assignment~$(\alpha,\beta)$ satisfies~$\phi$. By construction,
	$A \cup C \cup \{v_a,v_b\} \subseteq N^2   [\alpha]$ for every~$\alpha \in A$.
	Furthermore, for every~$\beta \in B \cap N^2[\alpha]$ we know that the truth assignment~$(\alpha,\beta)$
	does \emph{not} satisfy~$\phi$. We therefore can reformulate the condition
	under which a satisfying truth assignment~$(\alpha,\beta)$ does exist: if for any
	$\alpha \in A$ we have that~$|N^2[\alpha]| < |A| + |B| + |C| + 2 = 2N + m + 2$,
	then there must be a some~$\beta \in B \setminus N^2[\alpha]$, and thus
	$(\alpha,\beta)$ is satisfying. Note that the reverse holds as well: if there
	is a satisfying assignment for~$\phi$, the respective restrictions to~$X_l$
	and~$X_h$ are vertices in~$G$ with the aforementioned property.

	Assume that we can solve~\Problem{Closed $2$-Neighbourhood Sizes} for~$G$ in
	time~$O(|G|^{2-\eps})$ for some~$\eps > 0$. Since the output
	consists of~$|G|$ numbers, we can test in time~$O(|G| \log |G|)$ whether
	some vertex in~$A$ has strictly less than~$2N+m+2$ $\leq 2$-neighbours. But then
	we could find a satisfying assignment for~$\phi$ in time
	\[
		O\big(Nm + |G|^{2-\eps} + |G|\log |G|\big) = O\big(Nm + (2N+m+2)^{2-\eps} \big)
		= 2^{n \cdot (1-\eps/2)} m^{O(1)},
	\]
	contradicting SETH.
\end{proof}

\begin{figure}
	\centering
	\includegraphics[scale=.5]{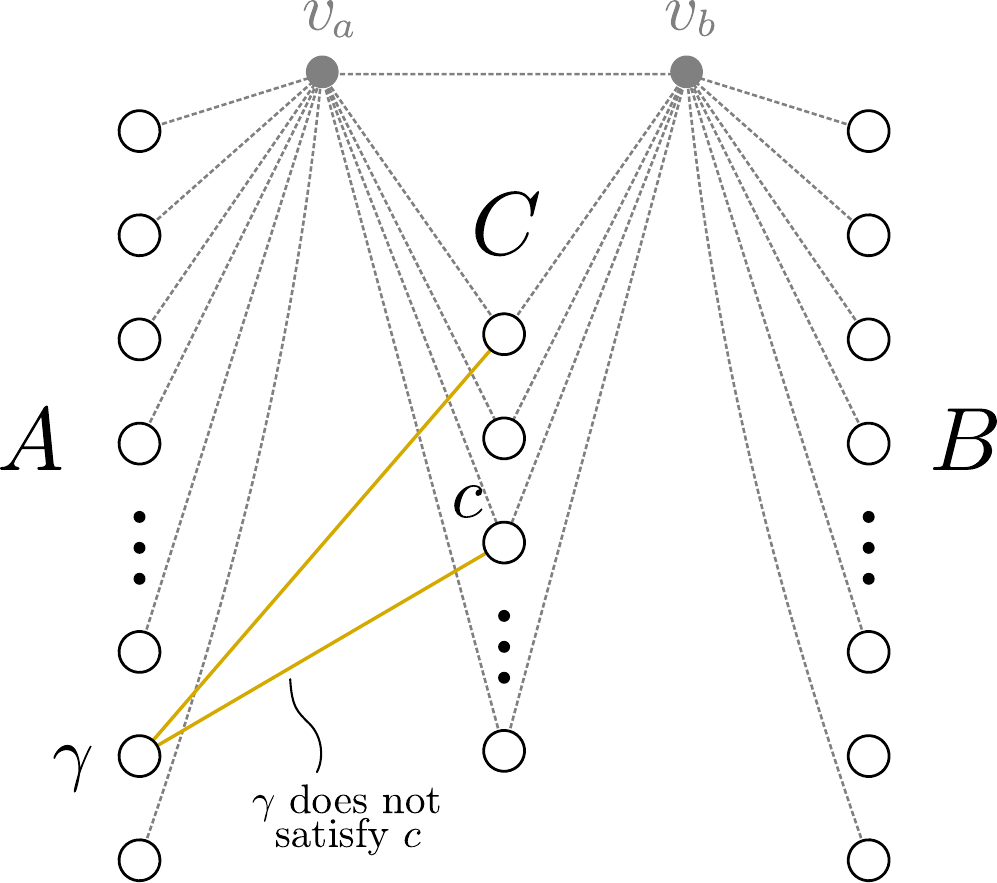}\hspace{3em}%
	\caption{%
		The reduction from \Problem{SAT} to \Problem{Closed 2-Neighbourhood Sizes}.
	}
	\label{figure-quadratic-lower-bound}
\end{figure}

%
%

\begin{theorem}\label{lemma:vc-lower-bound}
	For any~$\delta < 2$, \Problem{Closed $2$-Neighbourhood Sizes}
	cannot be solved in time~$O(2^{o(\vc(G))} \, n^{\delta})$, unless SETH fails.
\end{theorem}
\begin{proof}
  Let~$\phi$ be an~$r$-CNF formula and let~$\eps > 0$ be some constant we will
  fix later. Using the sparsification lemma, we construct~$t \leq 2^{\eps n}$
  formulas~$(\psi_i)_{i \leq t}$, each on~$n_i \leq n$ variables and~$m_i =
  O((\frac{r}{\eps})^{3r} n)$ clauses. For each~$\psi_i$ in turn, we apply the
  reduction from Lemma~\ref{lemma:no-subquadratic}. Notice that the resulting
  graph~$G$ has~$C \cup \{v_a,v_b\}$ as a vertex cover and thus
  $\vc(G) \leq |C|+2 = O((\frac{r}{\eps})^{3r} n)$.

  Assume towards a contradiction that we can solve \Problem{Closed
  $2$-Neighbourhood Sizes} in time $O(2^{o(\vc(G))} \, N^\delta)$.
  By inspecting the output of this hypothetical
  algorithm, we can determine again in time~$O(N \log N)$ whether~$\psi_i$ is satisfiable. The total running time of this
  algorithm would therefore be
  \[
      O(2^{o(\vc(G))} \, N^\delta + N\log N)
    = O(2^{o((\frac{r}{\eps})^{3r} n)} 2^{\delta n/2} + 2^{n/2} n)
    = 2^{(o((\frac{r}{\eps})^{3r}) + \delta/2 )n} O(n).
  \]
  Thus deciding whether the original formula~$\phi$ is decidable would
  be possible in total time
  \[
    2^{\eps n} n^{O(1)} + 2^{\eps n} \cdot 2^{(o((\frac{r}{\eps})^{3r}) + \delta/2 )n} O(n)
    = 2^{(o((\frac{r}{\eps})^{3r}) + \delta/2 + \eps)n} n^{O(1)}.
  \]
  For appropriate choices of~$\eps$, we can ensure that asymptotically
  $o((\frac{r}{\eps})^{3r}) + \delta/2 + \eps < 1$. But then the resulting
  algorithm contradicts the SETH and we conclude that the statement of 
  the theorem holds.

  Finally, let us note that
  the vertex cover does not need to be provided as input for the lower bound
  to hold since we can find it in time $O(1.2738^{\vc(G)} + mN)$~\cite{VertexCover}
  which is contained in~$O(2^{0.35\vc(G)}N)$.
\end{proof}

\noindent
The above construction implies several other algorithmic results, following from
the fact that~$\wcolnum_2(G) \leq \td(G) \leq \vc(G) \leq m+2$, $\tw(G) \leq \pw(G) \leq \td(G)-1 \leq m+1$
and~$\topgrad_1(G) \leq \grad_1(G) \leq m+2$ (\cf Observation~\ref{obs:nabla-vc})
and~$\dir \Delta_2(G) \leq \vc(G)$ (Observation~\ref{obs:dtf-vc}).
\begin{corollary}\label{cor:variouspar}
	Unless either the SETH fails, \Problem{Closed $2$-Neighbourhood Sizes} cannot be
	solved in time~$O(2^{o(f(G))} \, n^\delta)$ for any~$\delta < 2$ and any 
	structural parameter
	$f \in \{ \wcolnum_2, \dir \Delta_2 \vc, \pw, \tw, \td, \topgrad_1, \grad_1\}$.
\end{corollary}

\section{Upper Bounds}

The main results of this section are that, for every graph~$G$ with $n$
vertices,
\Problem{Closed $2$-Neighbourhood Sizes} can be solved in~$O(2^{\vc(G)/2}
\vc(G)^2 \cdot n)$ time (see Theorem~\ref{theorem:vc-upper-bound}) and
in~$O(2^w w \cdot n)$ time (see Theorem~\ref{theorem:tw-upper-bound}), where
$\vc(G)$ is the size of the minimum vertex cover and $w$ is the treewidth of
$G$. Before we proceed with the proof of Theorem~\ref{theorem:vc-upper-bound},
we first introduce now some needed infrastructure.

We will use calligraphic letters like~$\mathcal H, \mathcal Q$ for set
families and use~$\Delta(\mathcal H) := \max_{H \in \mathcal H} |H|$ to
denote the maximum cardinality of sets in~$\mathcal H$.
A \emph{weighted set family} over a universe set~$U$ is a tuple~$(\mathcal H, w)$
where~$\mathcal H$ is a family of sets over~$U$ and~$w \colon \mathcal H \to
\mathbb R^+$ assigns a positive rational weight to each member $H\in\mathcal H$.

\begin{definition}[Weighted set queries]
	Let~$(\mathcal H, w)$ be a weighted set family over the universe~$U$. We define the
	 following weighted queries for every~$S \subseteq U$:
	\[
			w_{\subseteq}(S) := \sum_{\substack{H \in \mathcal H \colon \\ S \subseteq H}} w(H), \quad
			w_{\supseteq}(S) := \sum_{\substack{H \in \mathcal H \colon \\ H \subseteq S}} w(H), \quad
			w_{\cap}(S) 		:= \sum_{\substack{H \in \mathcal H \colon \\ H \cap S \neq \emptyset}} w(H).
	\]
\end{definition}

\noindent
In other words, $w_{\subseteq}(S)$ (resp.~$w_{\supseteq}(S)$) returns total weight of sets in~$\mathcal H$
that are supersets (resp.~subsets) of~$S$ and~$w_{\cap}(S)$ the total weights of sets in~$\mathcal H$
that intersect~$S$.

We will assume in the following that functions with domain in~$2^U$ are
implemented as data structures which allow constant-time lookup and
modification. 
This can be done either in a randomized way via hash-maps 
or using the following deterministic implementation on a RAM: assuming that~$U =[n]$ for
some natural number~$n$, we store the value to a key $S =
\{s_1,s_2,\dots,s_p\}$ at address $s_1 + n\cdot s_2 + \dots + n^p \cdot s_p$.
The largest address used in this manner has size polynomial in~$n$, however,
we only need to initialize as many registers as we store values (which will
be linear in the following applications).

\begin{lemma}\label{lemma:queries}
	Given a weighted set family~$(\mathcal H, w)$ and a set family~$\mathcal Q$
	over~$U$ one can compute all values~$w_{\subseteq}(Q)$, $w_{\cap}(Q)$, $w_{\supseteq}(Q)$ for~$Q \in \mathcal Q$
	in time 
	\[
		O\big(2^{\Delta(\mathcal Q)} \Delta(\mathcal Q) |\mathcal Q| + 2^{\Delta(\mathcal H)} |\mathcal H| \big).
	\]
\end{lemma}
\begin{proof}
	First, we can easily compute all values for~$w_{\subseteq}(S)$ for~$S \subseteq U$
	as follows: for every~$H \in \mathcal H$, we increment a counter for each subset
	$S \subseteq H$. The resulting data structure gives exactly~$w_{\subseteq}$ and it
	takes~$O(\sum_{H \in \mathcal H} 2^{|H|}) = O(2^{\Delta(\mathcal H)} |\mathcal H|)$ arithmetic
	operations on values of~$w$ to compute it.
	Similarly, we compute all values for~$w_{\supseteq}(Q)$ for~$Q \in \mathcal Q$ 
	in time~$O(2^{\Delta(\mathcal Q)} |\mathcal Q|)$ by simply following the definition.

	For~$Q \in \mathcal Q$ and~$S \subseteq Q$ we define the auxiliary
	weighted query
	\[
			w_Q(S) := \sum_{\substack{H \in \mathcal H \colon \\ H \cap Q = S}} w(H),
	\]
	that is, $w_Q(S)$ returns the total weight of sets in~$H \in \mathcal H$ whose
	intersection with~$Q$ is precisely~$S$.
	We can compute $w_Q$ from~$w_{\subseteq}$ via the following inclusion-exclusion formula:
	\[
		w_Q(S) := \sum_{S \subseteq S' \subseteq Q} (-1)^{|S'\setminus S|} w_{\subseteq}(S').
	\]
	For a fixed~$Q \in \mathcal Q$, all values~$w_Q(S)$ for~$S \subseteq Q$ can be computed in
	time~$O(2^{|Q|} |Q|)$ using the fast \Mobius transform~\cite{YatesAlgorithm}. Given~$w_Q$,
	we can then compute~$w_{\cap}(Q)$ using the identity
	\[
		w_{\cap}(Q) = \sum_{S \subseteq Q} w_Q(S),
	\]
	hence all values of~$w_{\cap}(Q)$, $Q \in \mathcal Q$ can be computed in time
	\[
		O(2^{\Delta(\mathcal Q)}\Delta(\mathcal Q) |\mathcal Q| + \sum_{Q \in \mathcal Q} 2^{|Q|}) = O(2^{\Delta(\mathcal Q)} |\mathcal Q| ).
	\]
	Summing up the time needed to compute~$w_{\subseteq}$, $w_{\supseteq}$ and~$w_{\cap}$ yields the claimed bound.
\end{proof}

\noindent
As observed above, we can find a vertex cover of size~$t$ in time
$O^*(1.274^{t})$~\cite{VertexCover}, and thus the following result holds even
if the vertex cover~$X$ is not provided as input. We will use the Iverson
bracket notation~$\Iver{\phi}$ in the proof of Theorem \ref{theorem:vc-upper-
bound}, which evaluates to~1 if~$\phi$ is a true statement and to~0 otherwise.

\begin{theorem}\label{theorem:vc-upper-bound}
	For every graph~$G$ with $n$ vertices, 
	\Problem{Closed $2$-Neighbourhood Sizes} can be solved in~$O(2^{\vc(G)/2}
	\vc(G)^2 \cdot n)$ time.
\end{theorem}
\begin{proof}
	Let~$X$ be a vertex cover of~$G$ containing~$t$ vertices (in particular
	$t = \vc(G)$ if~$X$ is a minimum vertex cover).
	Let~$I := V(G)\setminus X$ be the complement independent set to~$X$.

	First, we compute the second neighbourhood size for vertices in
	the vertex cover~$X$ in time~$O(t n)$. To that end, let~$\tilde E$ contain
	all pairs~$u,v$ with~$u \in X$ and~$v \in X \cup I$ such that~$u,v$ have 
	exaclty distance two to each other. Since~$|\tilde E| \leq t^2 + t n = O(t n)$ we can
	easily compute $|\tilde E|$ in the claimed time. Consequently, we can also compute in $O(t n)$ time the
	size of the second neighbourhood for every vertex in~$X$.

	Let us now partition the independent set~$I$ into sets~$I_l, I_h$
	where~$I_l$ contains all vertices from~$I$ that have degree at most~$t/2$
	and~$I_h$ the remaining vertices. Note that every pair of vertices~$u,v \in I_h$
	will share at least one common neighbour in~$X$, hence all vertices in~$I_h$ have
	exactly distance two to each other---we can therefore count the contribution of
	$I_h$ to the  closed second neighbourhood of each member in~$I_h$ as~$|I_h|$.	
	It is therefore left to compute the contributions of vertices in~$I_l$ to
	vertices in~$I_l$, of vertices in~$I_h$ to vertices in~$I_l$, and of
	vertices in~$I_l$ to vertices in~$I_h$.

	To that end, let~$(\mathcal X_l, w^l)$ denote the weighted set family
	over~$X$ with~$\mathcal X_l := \{ N(v) \}_{v \in I_l}$ and where~$w^l(H)$
	simply counts the number of vertices in~$I_l$ that have neighbourhood
	exactly~$H$. Then for~$v \in I_l$ the value~$w^l_\cap(N(v))$ provides
	us exactly with the number of vertices in~$I_l$ whose neighbourhood intersects with~$N(v)$
	(including~$v$ itself), thus by Lemma~\ref{lemma:queries} we can compute the
	contribution of~$I_l$ to vertices in~$I_l$ in time 
	\[
		O(2^{\Delta(\mathcal X_l)} \Delta(\mathcal X_l) |\mathcal X_l|) = O(2^{t/2} t n).
	\]
	Next, to compute how vertices in~$I_l$ contribute to the second
	neighbourhood of vertices in~$I_h$, let~$\bar N(u) := X\setminus N(u)$ for every $u\in I_h$. We define the set family~$\bar{\mathcal X_h} := \{ \bar N(u) \}_{u \in I_h}$.	A vertex~$v \in I_l$ does
	\emph{not} contribute to the second neighbourhood of a vertex~$u \in I_h$
	if~$N(v) \subseteq \bar N(u)$. Of course, if we know the number of 
	vertices in~$I_l$ that do \emph{not} contribute to the second neighbourhood of~$u \in I_h$, we can easily compute the 
	number of vertices in~$I_l$ that contribute. For a given vertex~$u \in I_h$, we therefore want to compute
	\begin{align*}
		\sum_{v \in I_l} \Iver{N(v) \subseteq \bar N(u)}
		&= \sum_{H \in \mathcal X_l} \Iver{H \subseteq \bar N(u)} \cdot w^l(H) \\
		&= \,\, \sum_{\mathclap{\substack{H \in \mathcal X_l \colon \\ H \subseteq \bar N(u)}}} \, w^l(H) 
		= w^l_{\supseteq}(\bar N(u)),
	\end{align*}
	which, again, by Lemma~\ref{lemma:queries} can be computed for all sets in
	$\bar{\mathcal X_h}$ in time
	\[
		O(2^{\Delta(\mathcal X_l)} \Delta(\mathcal X_l) |\mathcal X_l| + 2^{\Delta(\bar{\mathcal X_h})} |\bar{\mathcal X_h}|)
		= O(2^{t/2} t n).
	\]
	Finally, let us compute how vertices in~$I_h$ contribute to the second
	neighbourhood of vertices in~$I_l$. Let~$(\bar{\mathcal X_h}, w^h)$ denote
	the weighted set family over~$X$ with~$\bar{\mathcal X_h}$ defined as above, where~$w^h(H)$  counts the number of vertices in~$I_h$ that have
	neighbourhood exactly~$X\setminus H$. To compute the contribution of~$I_h$
	to the second neighbourhood of a vertex~$v \in I_l$, we instead count the
	number of vertices in~$I_h$ that do \emph{not} contribute. Since~$N(u) \cap
	N(v) = \emptyset$ exactly when~$N(v) \subseteq X \setminus N(u)$, this
	quantity is given by
	\begin{align*}
		\sum_{H \in \bar{\mathcal X_h}} \Iver{N(v) \subseteq H} \cdot w^h(H)
		= \,\, \sum_{\mathclap{\substack{H \in \bar{\mathcal X_h} \colon \\ N(v) \subseteq H }}} \, w^h(H)
		= w^h_{\subseteq}(N(v)).
	\end{align*}
	We can compute all necessary values~$w^h_{\subseteq}(S)$ for~$S \in \mathcal X_l$ using Lemma~\ref{lemma:queries}
	in time
	\[
		O(2^{\Delta(\bar{\mathcal X_h})} \Delta(\bar{\mathcal X_h}) |\bar{\mathcal X_h}| + 2^{\Delta(\mathcal X_l)} |\mathcal X_l|)
		= O(2^{t/2} t n).
	\]
	Having computed all second neighbourhoods of vertices in~$X$ and summed up
	all possible ways in which vertices in~$I_h$ and~$I_l$ can contribute to
	each other's second neighbourhood in the claimed running time, we conclude the statement of the theorem. 
\end{proof}

In the next theorem we prove the second result of this section, namely an
algorithm for \Problem{Closed $2$-Neighbourhood Sizes} with exponential
dependency on the treewidth of the input graph.

\begin{theorem}\label{theorem:tw-upper-bound}
  For every graph~$G$ with $n$ vertices and with a tree decomposition of
  width~$w$ given as input,
\Problem{Closed $2$-Neighbourhood Sizes} can be solved in~$O(2^w w \cdot n)$ time.
\end{theorem}
\begin{proof}
  We can assume without loss of generality that the provided tree decomposition~$(X_i)_{i \in T}$ is nice (see Definition~\ref{nice-tree-decom-def}).
  For a bag~$X_i$ in the decomposition, we define the \emph{past} as~$P_i := \big( \bigcup_{j \in
  T_i} X_j \big) \setminus X_i$, and the \emph{future} as~$F_i := V(G)
  \setminus (P_i \cup X_i)$.

  We now pass over the decomposition in a bottom-up manner to compute a collection
  of dictionaries $N^P_i$, $i \in T$ with the following semantic: for every subset
  $Y \subseteq X_i$ we have that~$N^P_i[Y] := |N_G(Y) \cap P_i|$. Using backtracking,
  we afterwards compute the dictionary~$N^F_i$ with~$N^F_i[Y] := |N_G(Y) \cap F_i|$.
  Note that for a join-bag~$X_h$ with children~$X_i$, $X_j$ we have that
  \[
    N^P_h[Y] = N^P_i[Y] + N^P_j[Y], \quad Y \subseteq X_h,
  \]
  which follows easily from the assumption that the tree decomposition~$(X_i)_{i \in T}$ is nice.

  In our next pass over the decomposition, we keep track of the
  quantities~$|N(u) \cap P_i|$ and~$|N^2(u) \cap P_i|$ as well as the
  sets~$N(u) \cap X_i$ and the number of 2-neighbours that~$u$ has in~$G[X_i
  \cup P_i]$ for every~$v \in X_i$. Maintaining appropriate dynamic
  programming tables is simple for introduce- and forget operations, and the
  join-case is again a simple addition of table entries.

  Consider a vertex~$u$ and let~$X_i$ be the highest bag in which~$u$ appears,
  i.e.~either $X_i$ is the root bag and contains vertex $u$, or the parent
  bag~$X_j$ of~$X_i$ satisfies~$X_j := X_i \setminus \{u\}$. From the dynamic
  programming table in~$X_i$ we know the size of~$N^2(u) \cap P_i$. Now it holds
  that
  \begin{align*}
    |N^2(u)| &= |N^2(u) \cap P_j| + |N^2(u) \cap X_j  + |N^2(u) \cap F_j| \\
             &= |N^2(u) \cap P_i| + |N(N(u) \cap X_i) \cap X_i| + N^F_j[N(u) \cap X_i]
  \end{align*}
  and we have all three quantities readily available. To compute the closed second
  neighbourhood, we add the degree of~$u$. After passing over the whole tree decomposition
  we therefore have the size of every closed second neighbourhood of every vertex.

  The first pass over the decomposition takes time~$O(2^w n)$, the second one maintains tables of size~$O(w^2)$ and computes the degree inside bags in time~$O(w^2)$,
  for a total running time of~$O(w^2 n)$. We conclude that the total time
  taken is~$O(2^w n)$, as claimed.
\end{proof}

%

\section*{Conclusion}

\noindent
We used the SETH reduction toolkit by Abboud, Williams, and Wang to show that
computing the 2-neighbourhood sizes is neither possible in subquadratic time
nor in fpt-time with subexponential dependence on a range of `sparseness parameters'.
In that sense, the algorithm by Demaine \etal cannot be improved substantially; 
although a better exponential dependence of course remains possible.
We supplemented these lower bounds with algorithms that solve the problem in
time~$O(2^{\vc(G)/2} \vc(G)^2 \cdot n)$ and $O(2^w w \cdot n)$.

\end{document}

%% file: header.tex
\usepackage{bbm}
\usepackage{mathrsfs}
\usepackage{latexsym}
\usepackage{paralist}
\usepackage{microtype}

\usepackage{suffix}

\usepackage{xpunctuate}
\usepackage[all,british]{foreign} 
\usepackage{xfrac} 
\usepackage[style=english,english=british]{csquotes}

\usepackage{amsmath, amsthm, amssymb, amsfonts} 
\usepackage{thmtools, thm-restate}
\usepackage{mathtools}

\usepackage[full,small]{complexity}

\usepackage[final,notref,color]{showkeys} 

\usepackage{xargs}
\usepackage{afterpage} 

\PassOptionsToPackage{dvipsnames,usenames,table}{xcolor}
\usepackage{tikz}
\usetikzlibrary{calc}

\usepackage{booktabs} 
\usepackage{longtable}
\usepackage{float}
\usepackage{wrapfig}
\usepackage{floatflt}
\usepackage{framed}
\usepackage{ctable}
\usepackage{dcolumn}
\usepackage{adjustbox}

\usepackage{enumerate}
\usepackage{enumitem}
\usepackage{xspace}
\usepackage{xifthen}
\usepackage{fixltx2e}
\usepackage{url}
\usepackage{scrtime}

\usepackage[longend,vlined]{algorithm2e}

\usepackage{xkvltxp}
\usepackage[footnote,draft,english,silent,nomargin]{fixme}
\newcommand{\todo}[1]{\fxfatal{\color{red}#1}}

\usepackage{index}

\makeindex
\newindex{sym}{sym.idx}{sym.ind}{Symbol index}
\newindex{prob}{prob.idx}{prob.ind}{Problem index}

\newcommand\mathindex[1]{\index[sym]{\ensuremath{{#1}}}}
\WithSuffix\newcommand\mathindex*[1]{\index*[sym]{\ensuremath{{#1}}}}

\makeatletter
\newcommand{\raisemath}[1]{\mathpalette{\raisem@th{#1}}}
\newcommand{\raisem@th}[3]{\raisebox{#1}{$#2#3$}}
\makeatother

\def\N{\mathbb{N}} 

\renewcommand{\cal}{\mathcal}

\newcommand\Problem[2][]{%
    \index[prob]{#2@\textsc{#2}}\textsc{#2}\xspace%
}

\newcommand\restr[2]{{
  \left.\kern-\nulldelimiterspace 
  #1 
  \vphantom{\big|} 
  \right|_{#2} 
  }}

\def\sminor^#1{
    \preccurlyeq_{\mathrlap{\mathsf{m}}}^{#1}
} 
\def\ssminor^#1{%
    \mathbin{\dot\preccurlyeq_{{\mathrlap{\mathsf{m}}}}^{#1}}\,
} 
\def\stminor^#1{
    \preccurlyeq_{\mathrlap{\mathsf{t}}}^{#1}
} 
\def\sstminor^#1{%
    \mathbin{\dot\preccurlyeq_{{\mathrlap{\mathsf{t}}}}^{#1}}\,
} 

\def\dir#1{\vec{#1}}
\def\grad_#1{\nabla\!_{#1}}
\def\sgrad_#1{{\dot\nabla}\!_{#1}}
\def\topgrad_#1{\widetilde \nabla\!_{#1}}
\def\stopgrad_#1{\dot{\widetilde\nabla}\!_{#1}}
\def\topomega_#1{\widetilde \omega_{#1}}

\def\colnum_#1{ \operatorname{col}_{#1} }
\def\wcolnum_#1{ \operatorname{wcol}_{#1} }
\def\adm_#1{ \operatorname{adm}_{#1} }

\renewcommand{\leq}{\leqslant}
\renewcommand{\ge}{\geqslant}
\renewcommand{\geq}{\geqslant}

\renewcommand{\epsilon}{\varepsilon}

\renewcommand{\emptyset}{\varnothing}

\newlength{\convarrowwidth}
\usepackage{stmaryrd}
\def\Iver#1{ \llbracket #1 \rrbracket }

\newcommand{\widthm}[1]{ \mathbf{#1} } 
\DeclareMathOperator{\tw}{ \widthm{tw} }
\DeclareMathOperator{\pw}{ \widthm{pw} }
\DeclareMathOperator{\td}{ \widthm{td} }

\DeclareMathOperator{\width}{ \widthm{width} } 

\def\YYYY{{Y_0 \uplus Y_1 \uplus \cdots \uplus Y_\ell}} 
\WithSuffix\def\YYYY'{{Y'_0 \uplus Y'_1 \uplus \cdots \uplus Y'_{\ell'}}}

\def\vc{\mathop{\mathbf{vc}}}

\usepackage{pifont}
\newcommand{\yaay}{\kern4pt \ding{51} \kern-8pt \ding{51}}%
\theoremstyle{plain}
\newtheorem{lemma}{Lemma}

\newtheorem{theorem}{Theorem}
\newtheorem{corollary}{Corollary}
\newtheorem{observation}{Observation}

\newtheoremstyle{case}
  {\topsep}   
  {\topsep}   
  {}  
  {\parindent}       
  {\bfseries} 
  {\normalfont.}         
  {5pt plus 1pt minus 1pt} 
  {#1 #2: {\normalfont #3}}          

\theoremstyle{case}

\numberwithin{subcase}{case}
\makeatletter
\@addtoreset{case}{lemma}
\@addtoreset{case}{theorem}
\@addtoreset{case}{corollary}
\makeatother

\theoremstyle{definition}

\newtheorem{definition}{Definition}

\renewcommand*\etal{\xperiodafter{\emph{et~al}}}

\newcommand*\varrule[1][0.4pt]{\leavevmode\leaders\hrule height#1\hfill\kern0pt}

\setlist[1]{labelindent=\parindent,leftmargin=*} 
\setlist{itemsep=0pt}
\setitemize[1]{label={\small\textbullet}}

\newenvironment{tightcenter}
 {\parskip=0pt\par\nopagebreak\centering}
 {\par\noindent\ignorespacesafterend}

\usepackage{ctable}
\newlength{\RoundedBoxWidth}
\newsavebox{\GrayRoundedBox}
\newenvironment{GrayBox}[1]%
   {\setlength{\RoundedBoxWidth}{\textwidth-4.5ex}
    \def\boxheading{#1}
    \begin{lrbox}{\GrayRoundedBox}
       \begin{minipage}{\RoundedBoxWidth}%
   }{%
       \end{minipage}
    \end{lrbox}%
    \begin{tightcenter}%
    \begin{tikzpicture}%
       \node(Text)[draw=black!20,fill=white,rounded corners,%
             inner sep=2ex,text width=\RoundedBoxWidth]%
             {\usebox{\GrayRoundedBox}};
        \coordinate(x) at (current bounding box.north west);
        \node [draw=white,rectangle,inner sep=3pt,anchor=north west,fill=white] 
        at ($(x)+(6pt,.75em)$) {\boxheading};
    \end{tikzpicture}
    \end{tightcenter}\vspace{0pt}%
    \ignorespacesafterend
}    

\newenvironment{problem}[2][]{\noindent\ignorespaces%
                                \FrameSep=6pt%
                                \parindent=0pt%
                \vspace*{-.5em}
                \ifthenelse{\isempty{#1}}{%
                  \begin{GrayBox}{\textsc{#2}}%
                }{%
                  \begin{GrayBox}{\textsc{#2} parametrised by~{#1}}%
                }
                \index[prob]{#2@\textsc{#2}}%
                \newcommand\Prob{Problem:}%
                \newcommand\Input{Input:}%
                \begin{tabular*}{\textwidth}{@{\hspace{.1em}} >{\itshape} p{1.6cm} p{0.8\textwidth} @{}}%
            }{
                \end{tabular*}%
                \end{GrayBox}%
                \vspace*{-.5em}
                \ignorespacesafterend
            }  
\newenvironment{problem*}[2][]{\noindent\ignorespaces%
                                \FrameSep=6pt%
                                \parindent=0pt%
                \vspace*{-.5em}
                \ifthenelse{\isempty{#1}}{%
                  \begin{GrayBox}{\textsc{#2}}%
                }{%
                  \begin{GrayBox}{\textsc{#2} parametrised by~{#1}}%
                }
                \index[prob]{#2@\textsc{#2}|textbf}%
                \newcommand\Prob{Problem:}%
                \newcommand\Input{Input:}%
                \begin{tabular*}{\textwidth}{@{\hspace{.1em}} >{\itshape} p{1.6cm} p{0.8\textwidth} @{}}%
            }{
                \end{tabular*}%
                \end{GrayBox}%
                \vspace*{-.5em}
                \ignorespacesafterend
            }       

\renewenvironmentx{leftbar}[2][1=0.5pt, 2=5pt]{%
  \def\FrameCommand{\vrule width #1 \hspace{#2}}\MakeFramed {\advance\hsize-\width \FrameRestore}}%
{\endMakeFramed}

\newlength{\wleft}  \newlength{\wright}

\definecolor{Maroon}{cmyk}{0, 0.87, 0.68, 0.32}
\definecolor{RoyalBlue}{cmyk}{1, 0.50, 0, 0}
\definecolor{Black}{cmyk}{0, 0, 0, 0}
\definecolor{White}{rgb}{1, 1, 1}

\makeatletter
\let\orgdescriptionlabel\descriptionlabel
\def\@savelabel{}
\renewcommand*{\descriptionlabel}[1]{%
  \let\orglabel\label
  \let\label\@gobble
  \phantomsection
  \def\@savelabel{#1}
  \edef\@currentlabel{{\def\hfil{}#1}}
  \edef\@currentlabelname{#1}%
  \let\label\orglabel
  \orgdescriptionlabel{#1}%
}
\makeatother
\makeatletter
\def\namedlabel#1#2{\begingroup
   \def\@currentlabel{#1}%
   \label{#2}\endgroup
}
\makeatother

\usepackage{hyphenat}
\renewcommand{\th}{%
    \ifmmode
        ^\mathrm{th}%
    \else%
        \textsuperscript{th}\xspace%
    \fi%
}
\newcommand{\st}{%
    \ifmmode
        ^\mathrm{st}%
    \else%
        \textsuperscript{st}\xspace%
    \fi%
}
\newcommand{\nd}{%
    \ifmmode
        ^\mathrm{nd}%
    \else%
        \textsuperscript{nd}\xspace%
    \fi%
}
\newcommand{\rd}{%
    \ifmmode
        ^\mathrm{rd}%
    \else%
        \textsuperscript{rd}\xspace%
    \fi%
}

\def\Nesetril{Ne\v{s}et\v{r}il\xspace}

\def\Mobius{M\"{o}bius\xspace}

\usepackage{environ}%
\NewEnviron{nofiitable}{\noindent\ignorespaces%
  
  \[\arraycolsep=1.4pt%
  \begin{array}{rcr p{1cm} rcr}
    \BODY
  \end{array}
  \]
}

\newcolumntype{m}{>{$}l<{$}}
\newcolumntype{M}{>{$\displaystyle}l<{$}} 

\newcolumntype{L}{l}  
\newcolumntype{C}{c}  
\newcolumntype{R}{r}  
\newcolumntype{X}{>{\global\let\currentrowstyle\relax}}
\newcolumntype{^}{>{\currentrowstyle}}

%% file: main-arxiv.bbl
\begin{thebibliography}{1}

\bibitem{Abboud2016}
A.~Abboud, V.~Vassilevska~Williams, and J.~Wang.
\newblock Approximation and fixed parameter subquadratic algorithms for radius
  and diameter in sparse graphs.
\newblock In {\em Proceedings of the Twenty-seventh Annual ACM-SIAM Symposium
  on Discrete Algorithms}, SODA '16, pages 377--391, Philadelphia, PA, USA,
  2016. Society for Industrial and Applied Mathematics.

\bibitem{Sparsification}
C.~Calabro, R.~Impagliazzo, and R.~Paturi.
\newblock A duality between clause width and clause density for {SAT}.
\newblock In {\em 21st Annual {IEEE} Conference on Computational Complexity
  {(CCC} 2006), 16-20 July 2006, Prague, Czech Republic}, pages 252--260.
  {IEEE} Computer Society, 2006.

\bibitem{VertexCover}
J.~Chen, I.~A. Kanj, and G.~Xia.
\newblock Improved upper bounds for vertex cover.
\newblock {\em Theoretical Computer Science}, 411(40-42):3736--3756, 2010.

\bibitem{SparseComplexNetworks}
E.~D. Demaine, F.~Reidl, P.~Rossmanith, F.~S. Villaamil, S.~Sikdar, and B.~D.
  Sullivan.
\newblock Structural sparsity of complex networks: Bounded expansion in random
  models and real-world graphs.
\newblock {\em arXiv preprint arXiv:1406.2587}, 2014.

\bibitem{SETH}
R.~Impagliazzo and R.~Paturi.
\newblock Complexity of k-sat.
\newblock In {\em Computational Complexity, 1999. Proceedings. Fourteenth
  Annual IEEE Conference on}, pages 237--240. IEEE, 1999.

\bibitem{kloks1994}
T.~Kloks.
\newblock {\em Treewidth: Computations and Approximations}.
\newblock Lecture Notes in Computer Science. Rijksuniversiteir te Utrecht,
  1994.

\bibitem{Sparsity}
J.~\Nesetril and P.~{Ossona de Mendez}.
\newblock {\em Sparsity: Graphs, Structures, and Algorithms}, volume~28 of {\em
  Algorithms and Combinatorics}.
\newblock Springer, 2012.

\bibitem{YatesAlgorithm}
F.~Yates.
\newblock {\em The design and analysis of factorial experiments}.
\newblock Imperial Bureau of Soil Science, 1978.

\end{thebibliography}
